\documentclass[10pt,journal,draftclsnofoot,onecolumn]{IEEEtran}

\usepackage{array}

\usepackage{graphicx}
\usepackage[usenames]{color}

\usepackage{amsmath}  % Define \boldsymbol (in amsbsy too) and align
\usepackage{amssymb}  % Define \mathbb (in amsfonts too)

\usepackage{theorem}  % Helps in rendering theorems, etc.
\usepackage{cite}     % Gives multiple references as intervals
\usepackage{comment}  % Comments out with \begin{comment} ... \end{comment}

\usepackage{url}

\usepackage{amscd}
\usepackage{latexsym}

\usepackage{pict2e}

\title{Coding for the Lee and Manhattan Metrics \\ with Weighing Matrices}

\author{Tuvi Etzion, Fellow, IEEE, Alexander Vardy, Fellow, IEEE, and Eitan Yaakobi, Member, IEEE%
\thanks{T.~Etzion is with the Computer Science Department, Technion--Israel
Institute of Technology, Haifa 32000, Israel
(e-mail: \texttt{etzion@cs.technion.ac.il}).
Eitan Yaakobi is with the Electrical Engineering Department,
California Institute of Technology, Pasadena, CA 91125, U.S.A.
(e-mail: \texttt{yaakobi@caltech.edu}).
Part of this work was done when
he was with the
Department of Electrical and Computer~Engineering,
at the University of California San Diego,
La Jolla, CA 92093, U.S.A.
Alexander Vardy is with the
Department of Electrical and Computer~Engineering,
the Department of Computer Science and Engineering,
and the Depar\-tment of Mathematics,
all at the University of California San Diego,
La Jolla, CA 92093--0407, U.S.A.
(e-mail: \texttt{avardy@ucsd.edu}).}
\thanks{The material in this paper was presented in part in the 2010 IEEE Information
Theory Workshop, Dublin,
Ireland, August-September 2010.}
\thanks{This work was supported in part by the United States --- Israel
Binational Science~Foundation (BSF), Jerusalem, Israel, under Grant 2006097.}
}

\newcommand{\cH}{{\cal H}}

\newcommand{\cS}{{\cal S}}

\newcommand{\cW}{{\cal W}}

\newcommand{\bfc}{{\boldsymbol c}}

\newcommand{\bfp}{{\boldsymbol p}}

\newcommand{\bfs}{{\boldsymbol s}}

\newcommand{\bfv}{{\boldsymbol v}}
\newcommand{\bfx}{{\boldsymbol x}}
\newcommand{\bfy}{{\boldsymbol y}}
\newcommand{\bfz}{{\boldsymbol z}}
\newcommand{\bfw}{{\boldsymbol \omega}}

\newcommand{\bfe}{\mathbf{e}}

\newcommand{\deff}{\mbox{$\stackrel{\rm def}{=}$}}

\newcommand{\highsup}[1]{\raisebox{0.35ex}{\kern 1pt $\scriptstyle {#1} $}}

\makeatletter
\DeclareRobustCommand{\sbinom}{\genfrac[]\z@{}}
\makeatother
\newcommand{\G}[2]{\sbinom{{#1}\kern-1pt}{{#2}\kern-1pt}}
\newcommand{\Gq}[2]{\sbinom{{#1}\kern-1pt}{{#2}\kern-0.5pt}}

\newcommand{\F}{{\mathbb F}}

\newcommand{\C}{{\mathbb C}}
\newcommand{\dS}{{\mathbb S}}
\newcommand{\Z}{{\mathbb Z}}
\newcommand{\R}{{\mathbb R}}

\newcommand{\be}[1]{\begin{equation}\label{#1}}
\newcommand{\ee}{\end{equation}}

\newcommand{\bc}{\begin{center}}
\newcommand{\ec}{\end{center}}

\DeclareMathAlphabet{\mathbfsl}{OT1}{cmr}{bx}{it}
\newcommand{\uuu}{\kern-1pt\mathbfsl{u}\kern-0.5pt}
\newcommand{\vvv}{\kern-1pt\mathbfsl{v}\kern-0.5pt}

\newcommand{\sP}{\script{P}}

\newcommand{\sG}{\script{G}}

\newcommand{\Ps}{\smash{{\sP\kern-2.0pt}_q\kern-0.5pt(n)}}
\newcommand{\sPs}{\smash{{\sP\kern-1.5pt}_q(n)}}
\newcommand{\Gr}{\smash{{\sG\kern-1.5pt}_q\kern-0.5pt(n,k)}}
\newcommand{\Grmk}{\smash{{\sG\kern-1.5pt}_q\kern-0.5pt(n,n-k)}}
\newcommand{\Grdk}{\smash{{\sG\kern-1.5pt}_q\kern-0.5pt(2k,k)}}
\newcommand{\Grekappa}{\smash{{\sG\kern-1.5pt}_q\kern-0.5pt(n,e+1-\kappa)}}
\newcommand{\Grtwoekappa}{\smash{{\sG\kern-1.5pt}_q\kern-0.5pt(n,2e+1-\kappa)}}
\newcommand{\Gremkappa}{\smash{{\sG\kern-1.5pt}_q\kern-0.5pt(n,e-\kappa)}}
\newcommand{\Ptwo}{\smash{{\sP\kern-2.0pt}_2\kern-0.5pt(n)}}
\newcommand{\Ptwom}{\smash{{\sP\kern-2.0pt}_2\kern-0.5pt(m)}}
\newcommand{\Ptwonm}{\smash{{\sP\kern-2.0pt}_2\kern-0.5pt(n+m)}}
\newcommand{\Ptwoa}{\smash{{\sP\kern-2.0pt}_2\kern-0.5pt(1)}}
\newcommand{\Ptwob}{\smash{{\sP\kern-2.0pt}_2\kern-0.5pt(2)}}
\newcommand{\Ptwoc}{\smash{{\sP\kern-2.0pt}_2\kern-0.5pt(3)}}
\newcommand{\Ptwod}{\smash{{\sP\kern-2.0pt}_2\kern-0.5pt(4)}}
\newcommand{\Ptwoe}{\smash{{\sP\kern-2.0pt}_2\kern-0.5pt(5)}}
\newcommand{\Ptwokm}{\smash{{\sP\kern-2.0pt}_2\kern-0.5pt(2k-1)}}

\newcommand{\Gn}{\smash{{\sG\kern-1.5pt}_2\kern-0.5pt(n,n{-}1)}}
\newcommand{\Gnq}{\smash{{\sG\kern-1.5pt}_q\kern-0.5pt(n,n{-}1)}}
\newcommand{\Gone}{\smash{{\sG\kern-1.5pt}_2\kern-0.5pt(n,1)}}
\newcommand{\GTwo}{\smash{{\sG\kern-1.5pt}_2\kern-0.5pt(n,k)}}
\newcommand{\Gnk}{\smash{{\sG\kern-1.5pt}_2\kern-0.5pt(n,n{-}k)}}
\newcommand{\Pone}{\smash{{\sP\kern-2.5pt}_2\kern-0.5pt(n{-}1)}}
\newcommand{\Greone}{\smash{{\sG\kern-1.5pt}_q\kern-0.5pt(n,e{+}1)}}
\newcommand{\Gretwo}{\smash{{\sG\kern-1.5pt}_q\kern-0.5pt(n,e{+}2)}}

\newcommand{\zero}{{\mathbf 0}}

\newcommand{\myboxplus}{\kern1pt\mbox{\small$\boxplus$}}

\renewcommand{\leq}{\leqslant}

\renewcommand{\geq}{\geqslant}

\newcommand{\script}[1]{{\mathscr #1}}

\newcommand{\Cref}[1]{Co\-rol\-la\-ry\,\ref{#1}}

\theoremstyle{plain}
\theorembodyfont{\normalfont\slshape}

\newtheorem{thm}{Theorem$\!$}
\newenvironment{theorem}{\begin{thm}\hspace*{-1ex}{\bf.}}{\end{thm}}

\newtheorem{prop}[thm]{Proposition$\!$}

\newtheorem{lem}[thm]{Lemma$\!$}
\newenvironment{lemma}{\begin{lem}\hspace*{-1ex}{\bf.}}{\end{lem}}

\newtheorem{cor}[thm]{Corollary$\!$}

\newtheorem{conj}[thm]{Conjecture$\!$}

\newtheorem{remark}[thm]{Remark$\!$}

\newtheorem{defi}{Definition$\!$}

\theorembodyfont{\normalfont}

\newtheorem{exam}{Example$\!$}
\newenvironment{example}{\begin{exam}\hspace*{-1ex}{\bf .}}{\end{exam}}

\begin{document}

\maketitle

\begin{abstract}
This paper has two goals. The first one is to discuss good codes for packing problems
in the Lee and Manhattan metrics. The second one is to consider weighing matrices
for some of these coding problems.
Weighing matrices were considered as building blocks for
codes in the Hamming metric in various constructions.
In this paper we will consider mainly two types of weighing matrices,
namely conference matrices and Hadamard matrices, to construct
codes in the Lee (and Manhattan) metric. We will show
that these matrices have some desirable properties when considered
as generator matrices for codes in these metrics. Two related
packing problems will be considered. The first one is to find good codes
for error-correction (i.e. dense packings of Lee spheres).
The second one is to transform the space in a way that
volumes are preserved and each Lee sphere
(or conscribed cross-polytope), in the space,
will be transformed into a shape inscribed in a small cube.
\end{abstract}

%\vspace{.75ex}
%@@@@@@@@@@@@@@@@@@@@@@@@@@@@@@@@@@@@@@@@@@@@@@@@@@@@@@@@@@@@@@@@@@@@@@@@%
%                                                                        %
%         1. INTRODUCTION                                                %
%                                                                        %
%@@@@@@@@@@@@@@@@@@@@@@@@@@@@@@@@@@@@@@@@@@@@@@@@@@@@@@@@@@@@@@@@@@@@@@@@%

\section{Introduction}
%\vspace{-.5ex} \label{sec:introduction}

%\noindent\looseness=-1

\PARstart{T}{he} Lee metric was introduced
in~\cite{Lee58,Ulr57} for transmission of signals taken from
GF($p$) over certain noisy channels. It was generalized for $\Z_m$
in~\cite{GoWe}. The Lee distance $d_L (X,Y)$ between two words
$X=(x_1,x_2,\ldots,x_n)$, ${Y=(y_1,y_2,\ldots,y_n) \in \Z_m^n}$ is
given by $d_L(X,Y) \deff \Sigma_{i=1}^n \min \{ x_i -y_i~(\bmod~m) , y_i -x_i~
(\bmod~m) \}$. A related metric, the Manhattan metric, is defined
for alphabet letters taken from the integers. For two words
${X=(x_1,x_2,\ldots,x_n)}$, $Y=(y_1,y_2,\ldots,y_n) \in \Z^n$ the
Manhattan distance between $X$ and $Y$, $d_M(X,Y)$, is defined as ${d_M (X,Y) \deff
\Sigma_{i=1}^n | x_i - y_i |}$. A code~$\C$ in either metric has
minimum distance $d$ if for each two distinct codewords $c_1 , c_2
\in \C$ we have $d(c_1,c_2) \geq d$, where $d( \cdot , \cdot)$
stands for either the Lee distance or the Manhattan distance.

The main goal of this paper is to explore the properties of
some interesting dense codes in the Lee and
Manhattan metrics. Two related
packing problems will be considered.
The first one is to find good codes
for error-correction (i.e. dense packings of Lee spheres) in
the Lee and Manhattan metrics.
The second one is to transform the space in such a way that
volumes of shapes are preserved and each Lee sphere
(or conscribed cross-polytope), in the space,
will be transformed to a shape inscribed
in a small cube.
Some interesting connections between these two
problems will be revealed in this paper.

An {\it $n$-dimensional Lee sphere} $S_{n,R}$, with radius $R$, is
the shape centered at $(0, \ldots ,0)$ consisting of
all the points $(x_1 , x_2 , ... , x_n ) \in \Z^n$ which satisfy

$$\sum_{i=1}^n | x_i| \leq R~.$$

Similarly, an \emph{$n$-dimensional
cross-polytope} is the set consisting of all the points
$(x_1 , \ldots , x_n ) \in \R^n$ which satisfy the equation
$$
\sum_{i=1}^n | x_i | \leq 1~.
$$

A Lee sphere, $S_{n,R}$, centered at a point
$(y_1 , \ldots , y_n ) \in \Z^n$,
contains all the points of $\Z^n$ whose Manhattan distance from $(y_1 ,
\ldots , y_n )$ is at most $R$. The size of $S_{n,R}$ is well
known~\cite{GoWe}:
\begin{equation*}
| S_{n,R} | = \sum_{i=0}^{\min\{n,R\}} 2^i {n \choose i}{R \choose
i}
\end{equation*}\\
A code with minimum distance $d=2R+1$ (or $d=2R+2$) is a packing of
Lee spheres with radius $R$. Asymptotically,
the size of an $n$-dimensional Lee sphere with radius $R$
is $\frac{(2R)^n}{n!} + O(R^{n-1})$, when $n$ is fixed
and $R \longrightarrow \infty$.

The research on codes in the Manhattan metric
is not extensive.
It is mostly concerned with the existence and nonexistence of
perfect codes~\cite{GoWe,Post,Ast82B,Hor09}. Nevertheless, all
codes defined in the Lee metric over some finite alphabet,
(subsets of $\Z_m^n$) can be
extended to codes in the Manhattan metric over
the integers (subsets of $\Z^n$). The
literature on codes in the Lee metric is very extensive,
e.g.~\cite{RoSi94,ChWo71,GoWe,Sat79,Nak79,Ast82,Or93}.
Most of the interest at the beginning was in the existence
of perfect codes in these metrics.
The interest in
Lee codes increased in the last decade due to many new
applications of these codes. Some examples are
constrained and partial-response channels~\cite{RoSi94},
interleaving schemes~\cite{BBV}, orthogonal frequency-division
multiplexing~\cite{Sch07}, multidimensional
burst-error-correction~\cite{EtYa09}, and error-correction in the rank modulation
scheme for flash memories~\cite{JSB10}. The increased interest is also due
to new attempts to settle the existence question of perfect codes
in these metrics~\cite{Hor09}.

Linear codes are usually the codes which can be handled more
effectively and hence we will consider
only linear codes throughout this paper.

A linear code in $\Z^n$ is an integer lattice. A {\it lattice}
$\Lambda$ is a discrete, additive subgroup of the real $n$-space
$\R^n$,
\begin{equation}
\label{eq:lattice_def} \Lambda = \{ u_1 \bfv_1 + u_2 \bfv_2 + \cdots +
u_n \bfv_n ~:~ u_1,u_2, \cdots,u_n \in \Z \}~,
\end{equation}
where $\{ \bfv_1,\bfv_2,\ldots,\bfv_n \}$ is a set of linearly independent
vectors in $\R^n$. A lattice $\Lambda$ defined by
(\ref{eq:lattice_def}) is a sublattice of $\Z^n$ if and only if
$\{ \bfv_1,\bfv_2,\ldots,\bfv_n \} \subset \Z^n$. We will be interested
solely in sublattices of $\Z^n$. The vectors $\bfv_1,\bfv_2,\ldots,\bfv_n$
are called {\it basis} for $\Lambda \subseteq \Z^n$, and the $n
\times n$ matrix
$$
{\bf G}=\left[\begin{array}{cccc}
v_{11} & v_{12} & \ldots & v_{1n} \\
v_{21} & v_{22} & \ldots & v_{2n} \\
\vdots & \vdots & \ddots & \vdots\\
v_{n1} & v_{n2} & \ldots & v_{nn} \end{array}\right]
$$
having these vectors as its rows is said to be a {\it generator
matrix} for $\Lambda$. The lattice with generator matrix ${\bf
G}$ is denoted by $\Lambda ({\bf G})$.

%\begin{remark}
%There are other ways to describe a linear code in the
%Lee metric. The traditional way of using a
%generator matrix and its orthogonal
%parity-check matrix can be also used~\cite{RoSi94,Rot}. But, in
%our discussion, the lattice representation is the most convenient one.
%\end{remark}

The {\it volume} of a lattice $\Lambda$, denoted $V( \Lambda )$,
is inversely proportional to the number of lattice points per unit
volume. More precisely, $V( \Lambda )$ may be defined as the
volume of the {\it fundamental parallelogram} $\Pi(\Lambda)$,
which is given by
$$
\Pi(\Lambda) \deff\ \{ \xi_1 \bfv_1 + \xi_2 \bfv_2 + \cdots + \xi_n
\bfv_n~:~ 0 \leq \xi_i < 1,~ 1 \leq i \leq n \}
$$
There is a simple expression for the volume of $\Lambda$, namely,
$V(\Lambda)=| \det {\bf G} |$. An excellent reference,
for more material on lattices and some
comparison with our results, is~\cite{CoSl88}.

Sublattices of $\Z^n$ are periodic. We say that
the lattice $\Lambda$ has period $(m_1,m_2,\ldots,m_n) \in \Z^n$ if for
each $i$, ${1 \leq i \leq n}$, the point $(x_1,x_2,\ldots,x_n) \in
\Z^n$ is a lattice point in $\Lambda$ if and only if ${(x_1,\ldots,x_{i-1},x_i+m_i,
x_{i+1}, \ldots , x_n ) \in \Lambda}$. Let $m$ be the least common
multiple of the integers $m_1,m_2,\ldots,m_n$. The lattice $\Lambda$ has
also period $(m,m, \ldots , m)$ and it can be reduced
to a code $\C$ in the Lee metric over the alphabet~$\Z_m$.
It is easy to verify that the size of the code $\C$ is
$\frac{m^n}{V(\Lambda)}$. The
minimum distance of $\C$ can be the same as the minimum distance
of $\Lambda$, but it can be larger (for example, most
binary codes of length $n$ can be reduced from a sublattice of $\Z^n$,
where their Manhattan distance is at most 2. This is
the inverse of Construction A~\cite[p. 137]{CoSl88}).

One should note that if the lattice $\Lambda$ in the Manhattan metric
is reduced to a code over $\Z_p$, $p$ prime, in the Lee metric,
then the code is over a finite field. But, usually
the code in the Lee metric is over a ring which is not a field. It
makes its behavior slightly different from a code over a finite field.
Codes over rings were extensively studied in the last twenty years,
see e.g.~\cite{AhSo08,CaSl95,Car98,HKCSS}, and references therein.
In our discussion, a few concepts are important and for codes over~$\Z_m$ these are
essentially the same as the ones in traditional codes over a finite
field. For example, the minimum distance of the code
is the smallest distance between two codewords. The minimum distance
is equal to the weight of the word with minimum Manhattan (Lee) weight.

The definition of a coset for a lattice $\Lambda$ is very simple.
Let $\Lambda$ be a sublattice of $\Z^n$ and $\bfx \in \Z^n$.
The \emph{coset} of $\bfx$ is $\bfx + \Lambda \deff \{ \bfx+\bfc  \ | \
\bfc \in \Lambda \}$. The set of cosets is clearly unique.
For each coset we choose a \emph{coset leader}, which is a point in
the coset with minimum Manhattan weight. If there are a few points
with the same minimum Manhattan weight we choose one of them
(arbitrarily) as the coset leader. Once a set of coset leaders
is chosen then each point $\bfx \in \Z^n$ has a unique representation
as $\bfx = \bfc + \bfs$, where $\bfc$ is a lattice
point of $\Lambda$ and~$\bfs$ is
a coset leader. The number of different cosets
is equal to the volume of the lattice $\Lambda$.
In this context, the
covering radius of a lattice~$\Lambda$ (respectively a code $\C$)
is the distance of the word $\bfx$ whose
distance from the lattice (respectively code) is
the highest among all words. It equals to the weight
of the coset leader with the largest weight.
The covering radius of a lattice $\Lambda$ is the same as
the covering radius of the code
$\C$ reduced from $\Lambda$ to $\Z_m$, where $m$ is the period
of $\Lambda$.

A \emph{weighing matrix} $\cW_{n,w}$ of order $n$ and
weight $w$ is an $n \times n$ matrix over the alphabet
$\{ 0,1,-1 \}$ such that each row and column has
exactly $w$ nonzero entries; and $\cW \cdot \cW^T = wI_n$, where $I_n$ is the
identity matrix of order $n$.
The most important families of weighing matrices are the Hadamard matrices
in which $w=n$, and the conference matrices in which $w=n-1$.
In most of the results in this paper these families are considered.
Our construction in Section~\ref{sec:transform} will use
weighing matrices with some symmetry.
A weighing matrix $\cW$ is
\emph{symmetric} if $\cW^T = \cW$ and \emph{skew symmetric}
if $\cW^T = -\cW$\footnote{there is a generalization for this
definition for skew symmetric Hadamard matrices
(see~\cite[p. 89]{GeSe79}), but this generalization is
not considered in our paper.}. Information on weighing
matrices can be found for example in~\cite{CoDi07,GeSe79}.

In this paper we examine lattices and
codes related to weighing matrices. We prove that the
minimum Manhattan (respectively Lee) distance of the lattice (respectively code) derived
from a generator matrix taken as a weighing matrix of weight $w$,
is $w$. We discuss properties of Reed-Muller like codes,
i.e. based on Sylvester Hadamard matrices,
in the Lee and the Manhattan metrics.
These codes were used before for power control
in orthogonal frequency-division multiplexing
transmission. We prove bounds on their
covering radius and extend their range
of parameters. We define transformations
which transform $\R^n$ to $\R^n$ (respectively $\Z^n$ to $\Z^n$), in which
each conscribed cross-polytope (respectively Lee sphere) in $\R^n$ (respectively~$\Z^n$),
is transformed into a shape which can be
inscribed in a relatively small cube. The transformations
will preserve the volume of the shape and we believe that they are optimal
in the sense that there are no such transformations
which preserve volume and transform conscribed
cross-polytopes (respectively Lee spheres) into smaller cubes.
Generalization of the transformations yield some interesting lattices and codes which
are related to the codes based on Sylvester Hadamard matrices.

The rest of the paper is organized as follows.
In Section~\ref{sec:weighing} we discuss the use of weighing
matrices as generator matrices for codes (respectively lattices)
in the Lee (respectively Manhattan) metric. We will prove some properties
of the constructed codes (respectively lattices), their size, minimum distance, and on
which alphabet size they should be considered for the Lee metric.
In Section~\ref{sec:double} we will construct codes related to
the doubling construction of Hadamard matrices. We will discuss
their properties and also their covering radius.
In Section~\ref{sec:transform} we present the
volume preserving transformations which transform
each conscribed cross-polytope (respectively Lee sphere) in $\R^n$ (respectively $\Z^n$),
into a shape which can be
inscribed in a relatively small cube. These transformations are part
of a large family of transformations based on weighing
matrices and they will also yield some interesting codes.
Some connections to the codes obtained in Section~\ref{sec:double} will be discussed.
In Section~\ref{sec:exist}, the existence of weighing matrices
needed for our constructions will be discussed.
Conclusions and problems for future research are given in
Section~\ref{sec:conclude}.

\section{Weighing Matrices Codes}
\label{sec:weighing}

This section is devoted to codes whose generator
matrices are weighing matrices. We will discuss
some basic properties of such codes.

Each weighing matrix can be written in a \emph{normal form} such that its
first row consists only of \emph{zeroes} and +1's,
where all the \emph{zeroes} precede the +1's.
For this we only have to negate and permute columns.
We will now
consider weighing matrices written in normal form, unless
we will apply some operations on the original matrix
in normal form and obtain
one which is not in normal form.
We note that a weighing matrix $\cW_1$ is \emph{equivalent} to
a weighing matrix $\cW_2$ if $\cW_1$ is obtained from $\cW_2$
by permuting rows and columns, and/or negating rows and columns.

In the sequel, let $\mathbf{e}_i$ denote the unit vector with an \emph{one}
in the $i$-th coordinate, let $\bf{0}$ denote the all-zero vector,
and let $\bf{1}$ denote the all-one vector.

\begin{theorem}
\label{thm:LeeHad}
Let $\cW$ be a weighing matrix of order $n$ and weight~$w$
and let $\Lambda (\cW )$ be the corresponding lattice.
\begin{itemize}
\item The minimum Manhattan distance of $\Lambda (\cW )$ is $w$.

\item The volume of $\Lambda (\cW)$ is $w^{n \over
2}$.

\item $\Lambda (\cW)$ can be reduced to a code $\C$ of length $n$, in the Lee
metric, over the alphabet $\Z_w$. The minimum Lee distance of $\C$ is $w$.
\end{itemize}
\end{theorem}

\begin{proof}
\begin{itemize}
\item
The minimum distance of $\Lambda (\cW)$ is the
weight of the nonzero lattice point
of minimum Manhattan weight.

Let $\bfx \in \Lambda (\cW)$ be a nonzero
lattice point of minimum Manhattan weight.
The point $\bfx$ is obtained by a linear combination of a few
rows from $\cW$.

Let $\bfy$ be a row which is included in this linear
combination with a coefficient $\rho$, $\rho \geq 1$,
in this sum. Let $\cW'$ be the matrix obtained
from $\cW$ by negating the columns in which $\bfy$ has $-1$'s.

Let $\bfx'$ be a lattice point in $\Lambda (\cW')$
formed
from the linear combination of the same related
rows of $\cW'$, as those from which $\bfx$ was
formed from $\cW$.

The point $\bfx'$ has minimum Manhattan weight in $\Lambda (\cW')$
(the same Manhattan weight as $\bfx$).
The related row~$\bfy'$
has only \emph{zeroes} and \emph{ones}
and without loss of generality we can assume that the \emph{ones}
are in the last $w$ coordinates. $\bfy'$ is included
$\rho$ times in the linear combination from which $\bfx'$
is obtained. Hence, the total sum of the elements
in these last $w$ entries of $\bfx'$
is $\rho w$ (since in a row which is not $\bfy'$ the
sum of the entries in the last $w$ coordinates
is \emph{zero}). Thus, the Manhattan weight of $\bfx'$
is at least~$\rho w$.

All the rows in $\cW$ have Manhattan weight $w$
and hence the minimum Manhattan distance
of $\Lambda (\cW)$ is exactly~$w$.

\item The volume of $\Lambda (\cW)$ is the determinant of $\cW$ known to be
$w^{n \over 2}$ and this is easily inferred from the definition
of a weighing matrix.

\item Let $\bfw^{(i)}$ be row $i$ of $\cW$ and let
$\omega^{(i)}_j$ be the $j$-th entry in this row.
Clearly, $\sum_{i=1}^n \omega^{(i)}_j \bfw^{(i)} = w \cdot \mathbf{e}_j$
(since the $j$-th column is orthogonal to
all the columns of $\cW$ except
itself). Thus, $\Lambda (\cW)$ can be reduced to a code $\C$
of length $n$, in the Lee metric over, the alphabet $\Z_w$.
The lattice points of minimum Manhattan weight $w$ are also codewords
in $\C$ (where $-1$ is replaced by $w-1$)
and hence the minimum Lee distance of $\C$ is also $w$.
\end{itemize}
\end{proof}

A code is called \emph{self-dual} if it equals its dual.
Since the inner product of two rows from a weighing matrix
$\cW$ is either 0 or $w$, it follows that the code $\C$ reduced
from $\Lambda (\cW)$ is contained in its dual. Since the size of the
code is $w^{\frac{n}{2}}$ and the size of the space is $w^n$,
it follows that the dual code has also size~$w^{\frac{n}{2}}$.
Thus, we have
\begin{theorem}
\label{thm:selfdual}
Let $\cW$ be a weighing matrix of order $n$ and weight~$w$.
If $\C$ is the code over $\Z_w$ reduced from $\Lambda (\cW)$
then $\C$ is a self-dual code.
\end{theorem}

Self-dual codes were considered extensively in coding theory,
e.g.~\cite{RaSl98}.
The lattice $\Lambda'$ is the \emph{dual} of the lattice~$\Lambda$
if~$\Lambda'$ contains all the points in $\R^n$ whose inner product with
the lattice points of~$\Lambda$ is an integer. We are not interested in dual
lattices as the related
lattice points have usually some non-integer entries.

Let $A$ be an $n \times n$ matrix over $\Z_k$. The \emph{rank}
of $A$ over $\Z_k$ is defined to be the maximum number
of linearly independent rows of $A$ over $\Z_k$.

\begin{theorem}
\label{thm:rank}
The rank of a Hadamard matrix of order $n$ over $\Z_n$
is $n-1$.
\end{theorem}

\begin{proof}
Let $\cH$ be a Hadamard matrix of order $n$. By Theorem~\ref{thm:selfdual}
we have that
the code $\C$ reduced from $\Lambda (\cH)$, to the Lee metric
over $\Z_n$, is a self-dual code. Therefore, it is easy to verify
that the words of length $n$ with exactly two nonzero
entries equal $\frac{n}{2}$ are codewords of $\C$.
Consider the generator matrix
$$
G=\left[\begin{array}{cccccc}
\frac{n}{2} & 0 & 0 & \ldots & 0 & \frac{n}{2} \\
0 & \frac{n}{2} & 0 & \ldots & 0 & \frac{n}{2} \\
\vdots & \vdots & \vdots & \ddots & \vdots & \vdots\\
0 & 0 & 0 & \ldots  & \frac{n}{2} & \frac{n}{2} \\
0 & 0 & 0 & \ldots & 0 & n \end{array}\right]~;
$$
the lattice $\Lambda (G)$ is a sub-lattice of
$\Lambda (\cH)$. Thus, the rank of the code $\C$ reduced
from $\Lambda (\cH)$ has rank at least $n-1$.
If the rank of $\C$ is $n$ then it has at least one
row with Hamming weight one in each generator matrix. But this row
can only be $n \cdot \mathbf{e}_i$, for some $i$, $1 \leq i \leq n$.
It implies that the
rank of $\C$ is less than $n$, i.e. $n-1$.
\end{proof}

\begin{theorem}
\label{thm:cnf_p}
If $\cW$ is a conference matrix of order $n=p+1$,
where $p$ is a prime, then its rank over $\Z_p$ (also $\F_p$) is
$\frac{p+1}{2}$.
\end{theorem}

\begin{proof}
Since the volume of $\Lambda (\cW)$ is $p^\frac{p+1}{2}$ and $p$
is a prime, it follows that in a
lower diagonal matrix representation, the generator matrix of $\Lambda (\cW)$
has a diagonal with $\frac{p+1}{2}$ $p$'s and $\frac{p+1}{2}$ 1's.
Thus, the rank of $\cW$ over~$\Z_p$ is~$\frac{p+1}{2}$.
Since $p$ is a prime, it follows that
the ring~$\Z_p$ is equal to the finite field $\F_p$.
\end{proof}

\begin{conj}
\label{con:conf_MDS}
If $\C$ is a code of length $p+1$ constructed from a generator matrix
which is a conference matrix then $\C$ is an MDS code of dimension
$\frac{p+1}{2}$ and minimum Hamming distance~$\frac{p+3}{2}$.
\end{conj}
Conjecture~\ref{con:conf_MDS} was verified to
be true up to $n=23$, where the conference matrices are based on the
Paley's construction from quadratic residues modulo $p$.
Codes with these parameters (self-dual MDS of length $q+1$, $q$ a prime power)
were constructed in~\cite{GrGu08}.

\section{Codes from the Doubling Construction}
\label{sec:double}

The most simple and celebrated method to construct Hadamard
matrices of large orders from Hadamard matrices of small orders
is the \emph{doubling construction}. Given a Hadamard matrix $\cH$ of order $n$,
the matrix
$$
\left[\begin{array}{cc}
\cH & \cH \\
\cH & -\cH
\end{array}\right]~,
$$
is a Hadamard matrix of order $2n$.

A Sylvester Hadamard matrix of order $m$,
$\cH_m$, is a $2^m \times 2^m$ Hadamard
matrix obtained by the doubling construction starting with
the Hadamard matrix $\cH_0 =\begin{scriptsize}
\left[ \begin{array}{c} 1 \end{array} \right] \end{scriptsize}$ of order one. This matrix
is also based on the first order Reed-Muller code~\cite{McSl77}.
Let $H_0 = \begin{scriptsize} \left[ \begin{array}{c} 1 \end{array} \right] \end{scriptsize}$
and $H_{m+1} =\left[\begin{array}{cc}
H_m & H_m \\
0 & H_m
\end{array}\right]$, $m \geq 0$.
Let $G(m,j)$, $0 \leq j \leq m$, be the $2^m \times 2^m$ matrix constructed from
$H_m$ as follows. Let $2^\ell$ be the Hamming weight of the $s$-th row
of $H_m$. If $\ell \geq j$ then the $s$-th row of $G(m,j)$ will be
the same as the $s$-th row of $H_m$. If $\ell < j$ then the $s$-th row of
$G(m,j)$ will be the $s$-th row of $H_m$ multiplied by $2^{j-\ell}$.

It is easy to verify that $G(m,j)$ can be defined recursively as follows.
For $1 \leq j < m$, $G(m,j)$ is given by
$$
G(m,j) =\left[\begin{array}{cc}
G(m-1,j-1) & G(m-1,j-1) \\
0 & G(m-1,j)
\end{array}\right]~,
$$
where $G(m,m)$ is given by
$$
G(m,m) =\left[\begin{array}{cc}
G(m-1,m-1) & G(m-1,m-1) \\
0 & 2G(m-1,m-1)
\end{array}\right]~,
$$
and $G(m,0) = H_m$.

The following lemma can be proved by applying a simple induction.

\begin{lemma}
$$\Lambda (G(m,m))= \Lambda (\cH_m)~,$$
for all $m \geq 0$.
\end{lemma}

\begin{example}
The Sylvester Hadamard matrix of order 2, is a Hadamard
matrix of order 4, given by

$$
\cH_2 = \left[\begin{array}{cccc}
+1 & +1 & +1 & +1 \\
+1 & -1 & +1 & -1 \\
+1 & +1 & -1 & -1 \\
+1 & -1 & -1 & +1
\end{array}\right]~.
$$\\
In $\Z^4$ it generates the same lattice as the generator matrix
$$
G(2,2)= \left[\begin{array}{cccc}
1 & 1 & 1 & 1 \\
0 & 2 & 0 & 2 \\
0 & 0 & 2 & 2 \\
0 & 0 & 0 & 4
\end{array}\right]~.
$$\\
Reducing the entries of $G(2,2)$ into {\it zeroes} and {\it ones}
yields
$$
H_2 = \left[\begin{array}{cccc}
1 & 1 & 1 & 1 \\
0 & 1 & 0 & 1 \\
0 & 0 & 1 & 1 \\
0 & 0 & 0 & 1
\end{array}\right]~.
$$
\end{example}

\vspace{0.10cm}

Clearly, the rows of $G(m,j)$ are linearly independent.
Let $\Lambda(m,j)$ be the lattice whose generator matrix is $G(m,j)$,
and $\C(m,j)$ the code reduced from $\Lambda (m,j)$, over $\Z_{2^j}$,
whose generator matrix is $G(m,j)$.
$\C(m,j)$ was constructed by a completely
different approach for the control of the peak-to-mean envelope power ratio
in orthogonal frequency-division multiplexing in~\cite{Sch07},
where its size and minimum distance were discussed.
How this sequence of codes can be generalized
for length which is not a power of two and to Hadamard matrices which
are not based on Sylvester matrices?
A possible answer to this question and
our different approach for these codes will be demonstrated
in Section~\ref{sec:transform}.

The following lemma is an immediate result from the recursive
construction of $H_m$.
\begin{lemma}
\label{lem:NrowW}
The number of rows with weight $2^i$, $0 \leq i \leq m$, in $H_m$
is $\binom{m}{i}$.
\end{lemma}

By Lemma~\ref{lem:NrowW} and by the definition of $G(m,j)$,
for each~$\ell$, $j \leq \ell \leq m$,
there exist rows in $G(m,j)$ with Manhattan weight $2^\ell$.
These are the only weights of rows in $G(m,j)$.

\begin{theorem}
\label{thm:Gij}
\begin{itemize}
\item The minimum Manhattan distance of $\Lambda(m,j)$ is $2^j$.

\item The volume of the lattice $\Lambda (m,j)$ is
$\Pi_{i=0}^j 2^{(j-i) \binom{m}{i}}$.

\item $\Lambda(m,j)$ is reduced to the code $\C (m,j)$.
$\C(m,j)$ has minimum Lee distance $2^j$.
\end{itemize}
\end{theorem}

\begin{proof}
\begin{itemize}
\item The minimum distance of $\Lambda(m,j)$ can be derived by a simple induction
from the recursive definition of $G(m,j)$.

\item The volume of $\Lambda (m,j)$ can be derived easily by induction
from the recursive definition of $G(m,j)$ or by a very simple direct computation
from Lemma~\ref{lem:NrowW}.

\item It is easily verified by using induction that
for each $i$, $1 \leq i \leq n$, the point $2^j \cdot \mathbf{e}_i$
is contained in $\Lambda (m,j)$. Thus,
$\Lambda (m,j)$ can be reduced to a code $\C$ of length $2^m$, in the
Lee metric, over the alphabet $\Z_{2^j}$. The minimum Lee distance
can be derived also by a simple induction.
\end{itemize}
\end{proof}

In Section~\ref{sec:transform} we will consider codes
related to the lattice $\Lambda (m,j)$. The covering radius of these codes
will be an important factor in our construction
for a space transformation. Therefore, we will devote the rest of
this section to find bounds on the covering radius
of the lattice $\Lambda (m,j)$, which is equal to the covering radius of
the code $\C(m,j)$.

$\Lambda (m,0)$ is equal to $\Z^{2^m}$ and hence its covering
radius is~0. $\Lambda (m,1)$ consists of all the points in $\Z^{2^m}$
which have an even sum of elements. The covering radius
of this code is clearly~1.
$\C (m,2)$ is a diameter perfect code
with minimum Lee distance 4 and covering
radius 2~\cite{AAK01,Etz11}. In general
we don't know the exact covering radius of $\Lambda (m,j)$ except for
two lattices (codes) for which the covering radius was found with
a computer aid. The covering radius of $\Lambda (3,3)$
equals 6 and the covering radius of $\Lambda (4,3)$
equals 8. We also found that the covering radius of
$\Lambda (4,4)$ is at most 20. However, two bounds
can be derived from the structure of $G(m,j)$.
Let $r(m,j)$ be the covering radius of the lattice
$\Lambda (m,j)$ (and also the code $\C(m,j)$).
%We note that in the Hamming scheme the covering radius of related codes
%was considered in~\cite{HKM}.

\begin{theorem}
\label{thm:rmm}
$r(m,m) \leq 3 r(m-1,m-1) + 2^{m-1}$, $m \geq 5$, where
$r(2,2)=2$, $r(3,3)=6$ and $r(4,4) \leq 20$.
\end{theorem}

\begin{proof}
Let $(\bfx,\bfy) \in \Z_{2^m}^{2^m}$, where $\bfx,\bfy \in \Z_{2^m}^{2^{m-1}}$.
We have to show that there exists a codeword $\bfc \in \C(m,m)$
such that $d_L (\bfc,(\bfx,\bfy)) \leq 3 r(m-1,m-1) + 2^{m-1}$.

Let $\bfx' \in \Z_{2^{m-1}}^{2^{m-1}}$ be the word obtained from $\bfx$
by reducing each entry of $\bfx$ modulo $2^{m-1}$. Let $\bfz'$ be a codeword
in $\C(m-1,m-1)$ such that $d_L (\bfx',\bfz') \leq r(m-1,m-1)$.
By using the same linear combination of rows from
$G(m-1,m-1)$, which was used to obtain $\bfz'$
in $\Z_{2^{m-1}}^{2^{m-1}}$, with computation
modulo $2^m$ instead of modulo $2^{m-1}$ which
was used for $\bfz'$, we obtain a word
$\bfz \in \Z_{2^m}^{2^{m-1}}$ not necessarily equals to $\bfz'$
(for each coordinate, the values of $\bfz$ and $\bfz'$ are equal
modulo~$2^{m-1}$).

For each $i$,
$1 \leq i \leq 2^{m-1}$ we form a word $\bfp^{(i)} \in \Z_{2^m}^{2^m}$ as follows.
Let ${\ell_i = \min \{ z_i - x_i ~(\bmod~2^m),x_i - z_i ~(\bmod~2^m) \}}$, such that
$0 \leq \ell_i \leq 2^{m-1}$. If $\ell_i < 2^{m-2}$ then
$\bfp^{(i)} = (\zero , \zero )$. If $\ell_i \geq 2^{m-2}$ then
$\bfp^{(i)} = 2^{m-1} (\bfe_i ,\bfe_i )$.
It is easy to verify by the definition
of $G(m,m)$ that $(\bfz,\bfz) \in \C(m,m)$
and $\bfp^{(i)} \in \C(m,m)$ for each $1 \leq i \leq 2^{m-1}$.
Therefore, $(\bfv,\bfv)=(\bfz,\bfz) + \sum_{i=1}^{2^{m-1}} \bfp^{(i)}$
is a codeword in $\C(m,m)$ and $d_L (\bfx,\bfv)
= d_L (\bfx',\bfz') \leq r(m-1,m-1)$.

Let $\bfy' \in \Z_{2^m}^{2^{m-1}}$ be the
word defined as follows.
If $\bfv_i - \bfy_i$ is an even integer then
$\bfy'_i = \bfy_i$ and if $\bfv_i - \bfy_i$ is an odd
integer then $\bfy'_i = \bfy_i +1 (\bmod~2^m)$.
Since the covering radius of $\C(m-1,m-1)$ is $r(m-1,m-1)$,
it follows that there exists a codeword $\bfz'' \in \C(m-1,m-1)$ such
that $d_L (\bfy'-\bfv,2\bfz'') \leq 2r(m-1,m-1)$.
Clearly, $(\bfv,\bfv+2\bfz'') \in \C(m,m)$ and
$d_L((\bfx,\bfy'),(\bfv,\bfv+2\bfz'')) \leq 3r(m-1,m-1)$.
It implies that $d_L((\bfx,\bfy),(\bfv,\bfv+2\bfz'')) \leq 3 r(m-1,m-1)+2^{m-1}$.
Thus, $r(m,m) \leq 3 r(m-1,m-1) + 2^{m-1}$.
\end{proof}

One can analyze the bound of Theorem~\ref{thm:rmm} and obtain that
when $m$ is large $r(m,m)$ is less than
approximately $4 \cdot 3^{m-2}$, or $n^{1.585}$. But, we
believe that the covering radius of $\C(m,m)$ is considerably smaller.

\begin{theorem}
\label{thm:rmj}
$r(m,j) \leq r(m-1,j-1) + r(m-1,j)$, $2 < j < m$,
where $r(m,2)=2$ for $m \geq 2$ and
upper bound on $r(m,m)$ is given in Theorem~\ref{thm:rmm}.
\end{theorem}

\begin{proof}
Let $(\bfx,\bfy) \in \Z_{2^j}^{2^m}$, where $\bfx,\bfy \in \Z_{2^j}^{2^{m-1}}$.
We have to show that there exists a codeword $\bfc \in \C(m,j)$
such that $d_L (\bfc,(\bfx,\bfy)) \leq r(m-1,j-1) + r(m-1,j)$.

Let $\bfx' \in \Z_{2^{j-1}}^{2^{m-1}}$ be the word obtained from $\bfx$
by reducing each entry of $\bfx$ modulo~$2^{j-1}$.
Let $\bfz' \in \Z_{2^{j-1}}^{2^{m-1}}$ be a codeword
in $\C(m-1,j-1)$ such that $d_L (\bfx',\bfz') \leq r(m-1,j-1)$.
By using the same linear combination of rows from
$G(m-1,j-1)$, which was used to obtain $\bfz'$
in $\Z_{2^{j-1}}^{2^{m-1}}$, but with computation
modulo~$2^j$ instead of modulo~$2^{j-1}$ which
was used for $\bfz'$, we obtain a word
$\bfz \in \Z_{2^j}^{2^{m-1}}$ not necessarily equals to $\bfz'$
(for each coordinate, the values of $\bfz$ and $\bfz'$ are equal
modulo~$2^{j-1}$).

For each $i$,
$1 \leq i \leq 2^{m-1}$ we form a word $\bfp^{(i)} \in \Z_{2^j}^{2^m}$ as follows.
Let ${\ell_i = \min \{ z_i - x_i ~(\bmod~2^j),x_i - z_i ~(\bmod~2^j) \}}$, such that
$0 \leq \ell_i \leq 2^{j-1}$. If $\ell_i < 2^{j-2}$ then
$\bfp^{(i)} = (\zero , \zero )$. If $\ell_i \geq 2^{j-2}$ then
$\bfp^{(i)} = 2^{j-1} (\bfe_i ,\bfe_i )$.
It is easy to verify by the definition
of $G(m,j)$ that $(\bfz,\bfz) \in \C(m,j)$
and $\bfp^{(i)} \in \C(m,j)$ for each $1 \leq i \leq 2^{m-1}$.
Therefore, $(\bfv,\bfv)=(\bfz,\bfz) + \sum_{i=1}^{2^{m-1}} \bfp^{(i)}$
is a codeword in $\C(m,j)$ and $d_L (\bfx,\bfv)
= d_L (\bfx',\bfz') \leq r(m-1,j-1)$.

Since the covering radius of $\C(m-1,j)$ is $r(m-1,j)$,
it follows that there exists a codeword $\bfz'' \in \C(m-1,j)$ such that
$d_L (\bfy-\bfv,\bfz'') \leq r(m-1,j)$.
Clearly, $(\bfv,\bfv+\bfz'') \in \C(m,j)$ and
$d_L((\bfx,\bfy),(\bfv,\bfv+\bfz'')) \leq r(m-1,j-1)+r(m-1,j)$.
Thus, $r(m,j) \leq r(m-1,j-1) + r(m-1,j)$.
\end{proof}

%\begin{theorem}
%$$r(m,j) \leq \sum_{i=0}^{j-2} \binom{j+i-1}{i} \cdot 2 \cdot 3^{j-i-2}$$.
%\end{theorem}

The covering radius of a code can be computed also from
the parity-check matrix of the code. Hence,
it would be interesting
to examine the parity-check matrix of the code $\C(m,j)$.
We construct the parity check matrix $F(m,j)$ of the code
$\C(m,j)$ from the matrix $H_m$ as follows.
Let $2^\ell$ be the Hamming weight in the $s$-th row
of $H_m$. If $m-\ell < j$ then
$F(m,j)$ will contain the $s$-th row of $H_m$ multiplied by $2^{m-\ell}$.
There are no other rows in $F(m,j)$.

One can verify that $F(m,j)$, $0 \leq j \leq m$,
is defined recursively as follows (we leave the formal
proof to the reader). For $1 \leq j < m$, $F'(m,j)$ is
given by
$$
F'(m,j) =\left[\begin{array}{cc}
F'(m-1,j) & F'(m-1,j) \\
0 & 2F'(m-1,j-1)
\end{array}\right]~,
$$
where $F'(m,m)$ is defined by
$$
F'(m,m) =\left[\begin{array}{cc}
F'(m-1,m-1) & F'(m-1,m-1) \\
0 & 2F'(m-1,m-1)
\end{array}\right]~,
$$
and $F'(m,0) = [1~1~\cdots~1]$, for $m \geq 0$.

$F(m,j)$ is constructed from $F'(m,j)$ by omitting the last row.

\section{Lee Sphere Transformations}
\label{sec:transform}

In multidimensional coding, many techniques are applied on
multidimensional cubes of $\Z^n$ and cannot be
applied on other shapes in $\Z^n$,
e.g.~\cite{EtYa09,Abd86,BBZS}. Assume we want
to apply a technique which is applied
on any $n$-dimensional cube of $\Z^n$ to a different
$n$-dimensional shape $\cS$ of $\Z^n$. This
problem can be solved by a transformation from
$\Z^n$ to $\Z^n$, which preserves volumes, in which each
$n$-dimensional shape~$\cS$ of~$\Z^n$ is
transformed into a shape $\cS'$
which can be inscribed in a relatively
small $n$-dimensional cube of $\Z^n$. The
technique is now applied on the image of the transformation
and then transformed back into the
domain. One of the most important shapes in
this context is the $n$-dimensional Lee sphere with radius $R$,
$S_{n,R}$. Clearly, an $n$-dimensional Lee sphere with radius $R$ can be
inscribed in an $\underbrace{(2R+1)\times \cdots \times
(2R+1)}_{n\textrm{ times}}$ $n$-dimensional cube. In~\cite{EtYa09}
a transformation of $\Z^n$ is given for which $S_{n,R}$ is
transformed into a shape inscribed in a cube of size
$\underbrace{(R+1)\times (R+1)\times \cdots \times
(R+1)}_{n-1\textrm{ times}} \times(2R+1)$. The gap from the
theoretical size of the cube is still large since the size of the
$n$-dimensional Lee sphere with radius~$R$ is $\frac{(2R)^n}{n!} +
O(R^{n-1})$, when $n$ is fixed and $R \longrightarrow \infty$. The
goal of this section is to close on this gap. In the
process, some interesting codes and coding problems
will arise. The transformation
we have to define is clearly a discrete
transformation, but for completeness, and since it has
an interest of its own, we will consider also the
more simple case of a continuous transformation
$T:\mathbb{R}^n\rightarrow \mathbb{R}^n$. This can be viewed also
as a transformation on conscribed cross-polytopes, which
were defined in~\cite{GoWe}, rather than on Lee spheres.
For every Lee sphere,~$S_{n,R}$, the \emph{conscribed cross-polytope}, $CP_{n,R}$,
is defined~\cite{GoWe} to be the convex hull of the $2n$ centers points of
the $(n-1)$-dimensional extremal hyperfaces of $S_{n,R}$.
What makes this figure more attractive to us than similar figures is
that the volume of $CP_{n,R}$ is exactly $\frac{(2R+1)^n}{n!}$.

The continuous transformation
will also be interesting from error-correcting
codes point of view as it will be understood in the
sequel. The transformation which
will be described will make use of weighing matrices
with some symmetry (symmetric or skew-symmetric).
In Section~\ref{sec:exist} we will discuss the existence
of such matrices and their relevance in our construction.

\subsection{The Continuous Transformation}

In this subsection we are going to define a sequence
of transformations based on symmetric or skew symmetric weighing
matrices, some of which will transform Lee
spheres (or conscribed cross-polytopes) in the space,
into shapes inscribed in a relatively small cubes.
These transformations also
form some interesting codes in the Lee and Manhattan metrics
which are related to the codes defined in Section~\ref{sec:double}.
Of these transformations there is one which will preserve volumes
and will serve as our main transformation.

Let $\cW$ be a symmetric or skew symmetric weighing
matrix of order $n$ and weight $w$. Given a real number $s > 0$, we define a
transformation $T^\cW_s:\mathbb{R}^n\rightarrow \mathbb{R}^n$,
as follows. For each $\bfx = (x_1,\ldots,x_n)^t \in \mathbb{R}^n$,

\begin{equation}
\label{eq:TWs} T^\cW_s (\bfx) \deff \frac{\cW \bfx}{s}.
\end{equation}

\begin{lemma}
\label{lem:s_involutionW}
Let $\cW$ be a weighing matrix of order $n$ and weight $w$ and let $s >0$
be a positive real number.
\begin{itemize}
\item
If $\cW$ is symmetric then for all
${\bfx = (x_1,\ldots,x_n)^t \in \mathbb{R}^n}$,
$${T^\cW_{\frac{w}{s}} (T^\cW_s (\bfx)) =\bfx}.$$
\item
If $\cW$ is skew symmetric then for all
${\bfx = (x_1,\ldots,x_n)^t \in \mathbb{R}^n}$,
$${T^\cW_{\frac{w}{s}} (T^\cW_s (\bfx)) =-\bfx}.$$
\end{itemize}
\end{lemma}

\begin{proof}
We will prove the case where $\cW$ is a symmetric matrix.
For each $\bfx = (x_1,\ldots,x_n)^t \in \mathbb{R}^n$, $$T^\cW_{\frac{w}{s}}(T^\cW_s(\bfx)) =
\frac{\cW  \cdot \cW  \bfx}{\frac{w}{s} \cdot s} = \frac{\cW \cdot \cW^T \bfx}{w} = \frac{w I_n
\cdot \bfx}{w} = \bfx~.$$ The case in which $\cW$
is a skew symmetric matrix has an identical proof.
\end{proof}

\vspace{0.1cm}

Let $\cW$ be a symmetric or a skew symmetric
weighing matrix of order $n$
and weight $w$ and let $s$ be a positive integer which
divides $w$. Let $\Lambda^\cW_s$ be the set of points in $\Z^n$ which are
mapped to points of $\Z^n$ by the transformation~$T^\cW_s$
given by~(\ref{eq:TWs}), i.e.
$$
\Lambda^\cW_s \deff \{ \bfx \in \Z^n ~:~ T^\cW_s ( \bfx ) \in \Z^n \}~.
$$

The proof of the next theorem can be deduced from the theory of
dual lattices~\cite{CoSl88}. But, to avoid a new sequence of definitions
and known results we will provide another direct proof.
\begin{theorem}
\label{thm:kernel_transW} Let $\cW$ be a symmetric or a skew symmetric
weighing matrix
of order $n$ and weight $w$, and let $s$
be a positive integer which divides $w$. Then $\Lambda^\cW_s$ is a
lattice with minimum Manhattan
distance $s$; moreover $\Lambda^\cW_s
= T^\cW_{\frac{w}{s}}(\Lambda^\cW_{\frac{w}{s}})$.
Finally, $\Lambda^\cW_s$ can be reduced
to a code $\C^\cW_s$ of length~$n$, in the Lee
metric, over the alphabet $\Z_s$.
\end{theorem}

\begin{proof}
We break the proof into three parts. First, we will prove
that $\Lambda^\cW_s$ is a lattice. We will proceed to prove that
the minimum Manhattan distance of $\Lambda^\cW_s$ is $s$;
and that $\Lambda^\cW_s$ can be reduced to a code
$\C^\cW_s$ of length $n$, in the Lee
metric, over the alphabet $\Z_s$.
Finally, we will prove that
$\Lambda^\cW_s = T^\cW_{\frac{w}{s}}(\Lambda^\cW_{\frac{w}{s}})$.
\begin{enumerate}
\item
If $\bfx_1 , \bfx_2 \in \Lambda^\cW_s$ then $\bfx_1 , \bfx_2 \in \Z^n$ and
${T^\cW_s(\bfx_1), T^\cW_s(\bfx_2) \in \Z^n}$. Hence,
$$
T^\cW_s(\bfx_1 + \bfx_2) = \frac{\cW (\bfx_1 + \bfx_2 )}{s} =
\frac{\cW \bfx_1}{s} + \frac{\cW \bfx_2}{s}$$
$$ = T^\cW_s(\bfx_1)+T^\cW_s(\bfx_2) \in \Z^n ~,
$$
and therefore $\bfx_1 + \bfx_2 \in \Lambda^\cW_s$, i.e. $\Lambda^\cW_s$ is a lattice.

\item
Since $\Lambda^\cW_s$ is a lattice it follows that its minimum Manhattan
distance is the Manhattan weight of a nonzero lattice point
with minimum Manhattan weight.
Let $\bfx = (x_1,x_2,\ldots ,x_n)^t \in \Lambda^\cW_s$ be a nonzero
lattice point, i.e. ${(y_1,y_2,\ldots,y_n)^t = \frac{\cW \bfx}{s} =
( \frac{s \cdot y_1}{s}, \frac{s \cdot y_2}{s},
\ldots , \frac{s \cdot y_n}{s}) \in \Z^n}$.
There exists at least one $i$ for which $y_i \neq 0$. For this~$i$, we have
$s \cdot y_i = \sum_{j=1}^n \omega^{(i)}_j x_j$. Since $|w^{(i)}_j| \leq 1$
for every $j$, $1 \leq j \leq n$,
it follows that $\sum_{j=1}^n | x_j | \geq s$. Thus, the minimum Manhattan
weight of $\bfx$ is at least $s$ and the same it true
for the minimum Manhattan distance of~$\Lambda^\cW_s$.
It is easy to verify that $(0,\ldots,0,s,0,\ldots,0)^t$ is a point in $\Lambda^\cW_s$
and hence the minimum Manhattan distance of $\Lambda^\cW_s$ is exactly $s$; and
$\Lambda^\cW_s$ can be reduced to a code $\C^\cW_s$ of length $n$, in the Lee
metric, over the alphabet $\Z_s$.

\item
Let $\bfx \in \Lambda^\cW_s$, i.e. $\bfx \in \Z^n$, $\bfy = T_s(\bfx) \in Z^n$.
By Lemma~\ref{lem:s_involutionW} we have $T^\cW_{\frac{w}{s}}(\bfy) =
T^\cW_{\frac{w}{s}}(T^\cW_s(\bfx))$ equals either $\bfx$
or $-\bfx$, i.e. $\bfy \in \Lambda^\cW_{\frac{w}{s}}$,
which implies that $\bfx \in T^\cW_{\frac{w}{s}}(\Lambda^\cW_{\frac{w}{s}})$.
Therefore, $\Lambda^\cW_s \subseteq T^\cW_{\frac{w}{s}}(\Lambda^\cW_{\frac{w}{s}})$.

Let $\bfx \in T^\cW_{\frac{w}{s}}(\Lambda^\cW_{\frac{w}{s}})$, i.e.
$\bfx \in \Z^n$, $\bfx = T^\cW_{\frac{w}{s}} (\bfy)$,
where $\bfy \in \Lambda^\cW_{\frac{w}{s}} \subset \Z^n$.
By Lemma~\ref{lem:s_involutionW} we have that
$T^\cW_s (\bfx) = T^\cW_s(T^\cW_{\frac{w}{s}}(\bfy))$ equals either $\bfy$
or $-\bfy$, and hence $\bfx \in \Lambda^\cW_s$.
Therefore, $T^\cW_{\frac{w}{s}}(\Lambda^\cW_{\frac{w}{s}}) \subseteq \Lambda^\cW_s$.

Thus, $\Lambda^\cW_s=T^\cW_{\frac{w}{s}}(\Lambda^\cW_{\frac{w}{s}})$.
\end{enumerate}
\end{proof}

\begin{theorem}
\label{thm:dist_HADC}
Let $\cH$ be a Hadamard matrix of order $n > 4$ and let $s$
be a positive integer which divides $n$.
If $s$ is even then the minimum Lee distance of $\C^\cH_s$
is $s$.
If $s > 1$ is odd then the minimum Lee distance of $\C^\cH_s$
is greater than $s$ and is at most~$\frac{n}{2}$.
\end{theorem}

\begin{proof}
If $s$ is even and divides $n$ then $(0,\ldots,0,\frac{s}{2},\frac{s}{2})$
is a codeword in $\C^\cH_s$ and the minimum Lee distance $s$ follows from
Theorem~\ref{thm:kernel_transW}.

If $s$ is odd and divides $n$ then $s$ also
divides~$\frac{n}{2}$.
W.l.o.g. we assume that the first row of $\cH$ consists
of $n$ $+1$'s and the second row consists of
$\frac{n}{2}$ $+1$'s followed by $\frac{n}{2}$ $-1$'s.
Therefore, in all the other rows of $\cH$ we have exactly~$\frac{n}{4}$
$+1$'s and~$\frac{n}{4}$ $-1$'s
in the first $\frac{n}{2}$ entries. It implies
that the word of length~$n$ with~$\frac{n}{2}$ \emph{ones}
followed by~$\frac{n}{2}$ \emph{zeroes} is a codeword
in~$\C^\cH_s$. Hence, the minimum
Lee distance of~$\C^\cH_s$ is at most $\frac{n}{2}$.

Assume that the minimum Lee distance of $\C^\cH_s$ is $s$. Hence,
there is a lattice point $\bfx$ in $\Lambda^\cH_s$
with Manhattan weight $s$ and at least two
nonzero entries. Since $\cH$ contains a row with all entries +1's
it follows that the sum of elements in $\bfx$ is $-s$, 0, or $s$.
Since~$s$ is odd, it follows that this sum cannot be 0. Hence, $\bfx$ cannot
contain both positive and negative entries.
Since the projection, of the entries with nonzero
elements in $\bfx$, on $\cH$, has rows with both +1 and -1,
it follows that the inner product
of these rows with~$\bfx$ is not divisible by~$s$. Hence $\bfx$ is not a lattice point
in $\Lambda^\cH_s$, a contradiction. By Theorem~\ref{thm:kernel_transW},
the minimum Lee distance of~$\C^\cH_s$ is at least~$s$ and thus
the minimum Lee distance of~$\C^\cH_s$
is greater than~$s$ and is at most~$\frac{n}{2}$.
\end{proof}

\vspace{0.2cm}

Theorem~\ref{thm:dist_HADC} provides some information
on the minimum Lee distance of the code $\C^\cH_s$,
where $\cH$ is a Hadamard matrix.
In general, for a weighing matrix $\cW$,
what is the Lee distance of the code $\C^\cW_s$? It appears
that it is not always reduced to~$s$ as the Manhattan distance
of $\Lambda^\cW_s$.
In fact, if $\frac{n}{2} < w <n$ we conjecture that
it is always $w$, in contrast to the result in
Theorem~\ref{thm:dist_HADC} for $w=n$.

\begin{theorem}
\label{thm:CW}
If $\cW$ is a weighing matrix of order $n$ and weight $w$ then
$\Lambda^\cW_w = \Lambda (\cW)$.
\end{theorem}

\begin{proof}
Clearly, $\C^\cW_1$ is equal $\Z^n$.
By Theorem~\ref{thm:kernel_transW} we have that
$\Lambda^\cW_w = T^\cW_1 (\Lambda^\cW_1)$.  Since $\Lambda^\cW_1 =\Z^n$ it follows
that~$\Lambda^\cW_w$ contains exactly all
the linear combinations of the rows from~$\cW$.
Thus, $\Lambda^\cW_w =\Lambda (\cW)$.
\end{proof}

\begin{lemma}
\label{lem:TinC}
If $s_1$ divides $s_2$ and $s_1 < s_2$, then
$\Lambda^{\cW}_{s_2} \subset \Lambda^{\cW}_{s_1}$.
\end{lemma}

\begin{proof}
If $\bfx \in \Z^n$ and the entries of $\cW \bfx$ are divisible by $s_2$
then by definition we have $\bfx \in \Lambda^{\cW}_{s_2}$.
Since~$s_1$ divides~$s_2$, it follows that
the entries of $\cW \bfx$ are divisible also by $s_1$. Hence,
$\bfx \in \Lambda^{\cW}_{s_1}$ and
$\Lambda^{\cW}_{s_2} \subseteq \Lambda^{\cW}_{s_1}$.
By Theorem~\ref{thm:kernel_transW},
the minimum distance of $\Lambda^{\cW}_{s_1}$
is $s_1$ and the minimum distance
of $\Lambda^{\cW}_{s_2}$ is $s_2$.
Thus, $\Lambda^{\cW}_{s_2} \subset \Lambda^{\cW}_{s_1}$.
\end{proof}

\begin{cor}
If $s$ divides $w$ then $\Lambda^\cW_s$ contains $\Lambda (\cW)$.
\end{cor}

We now turn to a volume preserving transformation from the
set of all transformations which were defined.
This transformation is $T^{\cW}_{\sqrt{w}}$ and redefined
as ${T^\cW:\mathbb{R}^n\rightarrow \mathbb{R}^n}$ to be
\begin{equation}
\label{eq:TW} T^\cW(\bfx) \deff \frac{\cW \bfx}{\sqrt{w}}.
\end{equation}

%\begin{remark}
%If $w$ is a square then $T^\cW=T^\cW_{\sqrt{w}}$,
%but the transformation $T^\cW$ is
%also defined when $w$ is not a square.
%\end{remark}

%\begin{lemma}
%\label{lem:involutionW}
%Let $\cW$ be a weighing matrix of order $n$ and weight $w$
%\begin{itemize}
%\item
%If $\cW$ is symmetric then for all
%$\bfx = (x_1,\ldots,x_n)\in \mathbb{R}^n$,
%$$T^\cW (T^\cW (\bfx)) = \bfx.$$
%\item
%If $\cW$ is skew symmetric then for all
%$\bfx = (x_1,\ldots,x_n)\in \mathbb{R}^n$,
%$$T^\cW(T^\cW(\bfx)) = -\bfx.$$
%\end{itemize}
%\end{lemma}

Theorem~\ref{thm:kernel_transW} is applied also with the
transformation $T^{\cW}$. In this case
$w=D^2$, where $D$ is a positive integer,
$\Lambda^\cW \deff \Lambda^\cW_D$ is a lattice
with minimum Manhattan distance~$D$, and ${\Lambda^\cW
= T^\cW(\Lambda^\cW)}$. Finally, $\Lambda^\cW$
can be reduced to a code~$\C^\cW$ of length~$n$,
in the Lee metric, over the alphabet~$\Z_D$.

\begin{lemma}
\label{lem:Hinscribed}
A conscribed cross-polytope, centered at
$\bfc = (c_1,\ldots,c_n)^t \in \R^n$, $CP_{n,R} ( \bfc)$,
is inscribed after the transformation~$T^\cW$
inside an $n$-dimensional cube of size
$$\left( \frac{2R+1}{\sqrt{w}} \right) \times \cdots \times \left( \frac{2R+1}{\sqrt{w}} \right).$$
%Moreover, $\Big| T^\cW (CP_{n,R} ( \bfc )) \Big| = \Big| CP_{n,R} \Big|$.
\end{lemma}

\begin{proof}
$CP_{n,R} ( \bfc)$ is contained in the following set of points
$$
CP_{n,R} ( \bfc) =\left\{\bfc+ (x_1,\ldots,x_n)^t \ \Big| \ \sum_{i=1}^n|x_i|\leq R+\frac{1}{2} \right\}~,
$$
where $(x_1,\ldots,x_n)^t \in \R^n$. This set of points
is transformed by the transformation $T^\cW$ into the following set
of points,
\begin{align*}
& T^\cW(CP_{n,R} ( \bfc)) & \\
& = \left\{T^\cW(\bfc+(x_1,\ldots,x_n)^t) \ \Big| \ \sum_{i=1}^n|x_i|\leq R+\frac{1}{2} \right\} & \\
& = \left\{\frac{\cW \bfc}{\sqrt{w}} + \frac{\cW (x_1,\ldots,x_n)^t}{\sqrt{w}}\ \Big| \ \sum_{i=1}^n|x_i|\leq R+\frac{1}{2} \right\}. &
\end{align*}
If $\cW (x_1,\ldots,x_n)^t = (y_1,\ldots,y_n)^t$ then, for $1\leq i\leq n$,
$$
|y_i| =\Big| \sum_{j=1}^n \omega^{(i)}_j x_j \Big| \leq  \sum_{j=1}^n|x_j| \leq R+\frac{1}{2}~.
$$
Therefore, the set $T^\cW (CP_{n,R} ( \bfc))$ is located inside the following
$\left( \frac{2R+1}{\sqrt{w}} \right) \times \cdots \times
\left( \frac{2R+1}{\sqrt{w}} \right)$ $n$-dimensional cube
$$
\left\{\frac{\cW \bfc}{\sqrt{w}} + (\ell_1,\ldots ,\ell_n)^t \ \Big| \ |\ell_i| \leq \frac{R+\frac{1}{2}}{\sqrt{w}}, ~1 \leq i \leq n ~ \right\}~.
$$
%Finally, the transformation $T^\cW$ is an affine transformation
%and $\det(\cW/\sqrt{w}) = 1$ then it also preserves volume.
%Thus, $\Big| T^\cW (CP_{n,R} ( \bfc )) \Big| = \Big| CP_{n,R} \Big|$.
\end{proof}

Note that since $\det(\cW/\sqrt{w}) = 1$,
the transformation $T^\cW$ also preserves volumes.
The volume of the conscribed cross-polytope is $\frac{(2R+1)^n}{n!}$
while the volume of the inscribing {$n-$dimensional} cube is $\frac{(2R+1)^n}{\sqrt{w}^n}$.
If we choose $w=n$, i.e. a Hadamard matrix of
order $n$, then we get that the ratio between
the volumes of the $n$-dimensional cube and
the conscribed cross-polytope is $\frac{n!}{n^{n/2}}$.
The shape of the Lee sphere is very similar to the one
of the conscribed cross-polytope and hence a similar
result can be obtained for a Lee sphere. But, a continuous
shape like the conscribed cross-polytope is more
appropriate when we consider a continuous transformation.

%We continue to present a few properties of the codes related
%to the transformation $\C_s$ which will lead to a bound on the
%covering radius of the code derived from a weighing matrix $\cW$.
%
%\begin{lemma}
%If $\bfx$ is a codeword of $\C_s$ then the sum of entries
%of $\bfx$ is divisible by $s$.
%\end{lemma}
%\begin{proof}
%By definition of $T_s$ we have that
%a word $\bfx$ can be a codeword of $\C_s$ only if each entry
%of $\cH \bfx$ is divisible by $s$. Since $\cH$ contains
%an all-one row it implies that the sum of entries
%of $\bfx$ is divisible by $s$.
%\end{proof}

%
%One might be tempted to conjecture that the covering radius of the code
%$\C^\cH_s$, $2 < s < n$, is less than $s$. But, unfortunately, this is
%clearly not true in general.
%For example, let $n=12k$ and consider a Hadamard matrix $\cH$ of order $n$.
%Let $\bfx \in \Z^n$ be any word of Hamming weight three with exactly three
%entries equal 2. One can verify that $d_M (\bfx, \C_3)=3$ and hence the covering
%radius of $\C^\cH_3$ is at least 3. It is not difficult to find similar examples
%for other parameters.

\subsection{On the connection between $\C(m,j)$ and $\C_{2^j}^{\cH_m}$}

In this subsection we consider some interesting
connections between the code $\C(m,j)$ defined in Section~\ref{sec:double}
and the code $\C^{\cH_m}_{2^j}$ defined in
Theorem~\ref{thm:kernel_transW}, where the weighing matrix $\cW$ is
the Hadamard matrix $\cH_m$.

\begin{lemma}
\label{lem:divGm_j}
The inner product of a lattice point from $\Lambda(m,j)$ and a
lattice point from $\Lambda(m,m)$ is divisible by $2^j$.
\end{lemma}

\begin{proof}
The lattice $\Lambda (m,m)$ is equal to the
lattice $\Lambda(\cH_m)$.
By Theorem~\ref{thm:LeeHad}, this lattice is reduced to a code,
in the Lee metric, over the alphabet $\Z_{2^m}$. It follows that the
inner product between lattice points of $\Lambda(m,m)$ is divisible by~$2^m$.
The sum of the entries in a row of $G(m,m)$ is exactly~$2^m$.
The sum of elements in a given row of $G(m,j)$ is~$2^\ell$,
for some~$\ell$ such that $j \leq \ell \leq m$. This row is obtained
by dividing the entries of the related row in $G(m,m)$ by $2^{m-\ell}$.
Hence, the inner product of this row with any lattice point of $\Lambda(m,m)$
is divisible by~$2^\ell$. Therefore, the inner product of any row in
$G(m,j)$ and a lattice point from $\Lambda(m,m)$ is divisible by~$2^j$.

Thus, the inner product of a lattice point from $\Lambda(m,j)$ and a lattice
point from $\Lambda (m,m)$ is divisible by $2^j$.
\end{proof}

\begin{cor}
\label{cor:divGm_j}
The inner product of a codeword from $\C(m,j)$ and a
codeword from $\C(m,m)$ is divisible by $2^j$.
\end{cor}

\begin{lemma}
\label{lem:EinC}
$\C(m,j) \subseteq \C^{\cH_m}_{2^j}$.
\end{lemma}

\begin{proof}
By Corollary~\ref{cor:divGm_j}, the inner product between a codeword
of $\C(m,j)$ and a codeword of $\C(m,m)$ is divisible by~$2^j$.
Since $\cH_m$ is the generator matrix of $\C(m,m)$, it follows that
if $\bfx \in \C(m,j)$ then the entries of $\cH_m \bfx^t$ are divisible
by~$2^j$ and therefore $\bfx \in \C^{\cH_m}_{2^j}$.

Thus, $\C(m,j) \subseteq \C^{\cH_m}_{2^j}$.
\end{proof}

\begin{cor}
The covering radius of the code $\C^{\cH_m}_{2^j}$
is less than or equal to the
covering radius of the code $\C(m,j)$.
\end{cor}

\begin{conj}
\label{con:EequalC}
$\C(m,j) = \C^{\cH_m}_{2^j}$.
\end{conj}

For the next result, we need one more definition and a few observations.
For a given word $\bfx=(x_1,\ldots,x_n)$, the \emph{reverse} of $\bfx$, $\bfx^R$,
is the word obtained from $\bfx$ by reading its elements from the last to the
first, i.e. $\bfx^R \deff (x_n,\ldots, x_1)$. It is readily verified that
for each $j$, $1 \leq j \leq m$, $\bfx \in \C(m,j)$ if and only if
$\bfx^R \in \C(m,j)$. Moreover, if we take the matrix which consists of
all the reverse rows of $G(m,j)$ we will obtain another generator
matrix for $\C(m,j)$.

\begin{lemma}
\label{lem:which_mult}
If $\bfx$, the $i$-th row of the matrix $H_m$, has Manhattan weight $2^\ell$ then
$\bfx^R \cdot \cH_m$ is a multiple by $2^\ell$
of the reverse for the $(2^m+1-i)$-th row of the matrix $H_m$.
%(given in reverse order).
\end{lemma}

\begin{proof}
The proof is by induction on $m$. The trivial basis is $m=2$.
Assume the claim is true when the matrices involved are $H_m$ and
$\cH_m$. Let $\bfx = ( \tilde{\bfx} , \tilde{\bfx} )$ be the $i$-th row of $H_{m+1}$,
and let $\bfv$ be the $(2^m+1-i)$-th row of $\cH_m$,
$1 \leq i \leq 2^m$. If the weight of $\tilde{\bfx}$ is $2^\ell$
then the weight of $( \tilde{\bfx} , \tilde{\bfx} )$ is $2^{\ell+1}$.
$\tilde{\bfx}$ is the ${i-}$th row of $H_m$ and hence by
the induction hypothesis $\tilde{\bfx}^R \cdot \cH_m$ is a multiple by $2^\ell$ of
$\bfv^R$. It is also easy to verify that if $\bfz = \bfy \cdot \cH_m$ then
$( 2\bfz,\zero) = (\bfy,\bfy) \cdot \cH_{m+1}$. The $(2^{m+1}+1-i)$-th
row of $\cH_{m+1}$ is $(\zero,\bfv)$ and hence
$( \tilde{\bfx} , \tilde{\bfx} )^R  \cdot \cH_{m+1}$
is a multiple by $2^{\ell+1}$ of $(\bfv^R ,\zero)$.
\end{proof}

A simple proof by induction on $m$ using the structure of $H_m$ can
be given for the following lemma.
\begin{lemma}
\label{lem:weights_H}
If the $i$-th row of $H_m$ has weight $2^\ell$ then the $(2^m+1-i)$-th
row of $H_m$ has weight $2^{m-\ell}$.
\end{lemma}

\begin{lemma}
$T^{\cH_m}_{2^j} ( \Lambda (m,j)) = \Lambda (m,m-j)$.
\end{lemma}

\begin{proof}
If $\bfx$, the $i$-th row of $H_m$, has weight $2^\ell$, $\ell \geq j$, then
by definition the $i$-th row of $G(m,j)$ is equal~$\bfx$. By Lemma~\ref{lem:which_mult},
we have that $\bfx^R \cdot \cH_m$ is a multiple by $2^\ell$ of
the reverse of the $(2^m+1-i)$-th row of $H_m$. Hence $T^{\cH_m}_{2^j} (\bfx^R)$ is
a multiple by $2^{\ell-j}$ of the reverse of the $(2^m+1-i)$-th row of $H_m$.
By Lemma~\ref{lem:weights_H}, the
$(2^m+1-i)$-th row of $H_m$ has weight $2^{m-\ell}$. Therefore, the
Manhattan weight of $T^{\cH_m}_{2^j} (\bfx^R)$ is $2^{m-j}$ which implies
by the definition of $G(m,m-j)$ that $T^{\cH_m}_{2^j} (\bfx^R)$ is a row in $G(m,m-j)$.

If $\bfx$, the $i$-th row of $H_m$, has weight $2^\ell$, $\ell < j$, then
the $i$-th row of $G(m,j)$ is equal $2^{j-\ell} \bfx$ and its weight
is~$2^j$. By Lemma~\ref{lem:which_mult}
we have that $\bfx^R \cdot \cH_m$ is a multiple by $2^j$ of
the reverse of the $(2^m+1-i)$-th row of $H_m$.
Hence $T^{\cH_m}_{2^j} (2^{j-\ell} \bfx^R)$ is
the reverse of the $(2^m+1-i)$-th row of $H_m$.
By Lemma~\ref{lem:weights_H} the
$(2^m+1-i)$-th row of $H_m$ has weight $2^{m-\ell} > 2^{m-j}$. Therefore,
by the definition of $G(m,m-j)$ we
have that $T^{\cH_m}_{2^j} (2^{j-\ell} \bfx^R)$ is a row in $G(m,m-j)$.

We have shown that a basis of $\Lambda (m,j)$ is transformed by the
transformation $T^{\cH_m}_{2^j}$ to a basis of $\Lambda(m,m-j)$.
Since $\Lambda(m,j)$ and $\Lambda(m,m-j)$ are lattices,
and $T^{\cH_m}_{2^j}$ is a linear transformation,
it implies that $T^{\cH_m}_{2^j} ( \Lambda (m,j)) = \Lambda (m,m-j)$.
\end{proof}

\begin{cor}
$\C^{\cH_m}_{2^j} = \C(m,j)$ if and only if $\C^{\cH_m}_{2^{m-j}} = \C(m,m-j)$.
\end{cor}

\subsection{The Discrete Transformation}

For the discrete case we want to modify the
transformation~$T^\cW$, used for the continuous case.
Let $D$ be a positive integer and
$\cW$ a symmetric weighing matrix of
order $n$ and weight $w=D^2$. Let $\dS$ be the set of
coset leaders of the lattice $\Lambda^\cW$ defined in
Theorem~\ref{thm:kernel_transW} based on~(\ref{eq:TW}). The discrete transformation
${\bf \tilde{T}}^\cW:\mathbb{Z}^n\rightarrow \mathbb{Z}^n$ is defined
as follows. For each ${(x_1,\ldots,x_n)\in \mathbb{Z}^n}$,
let $(x_1,\ldots,x_n) = (c_1,\ldots,c_n) +
(s_1,\ldots,s_n)$, where $(c_1,\ldots,c_n)^t \in \Lambda^\cW$ and
${(s_1,\ldots,s_n)^t \in \dS}$.
The choice of the pair $(c_1,\ldots,c_n)$ and
$(s_1,\ldots,s_n)$ is unique once the set of coset leaders~$\dS$ is defined. Let
\begin{align*}
{\bf \tilde{T}}^\cW ((x_1,\ldots,x_n)^t) =
T^\cW ((c_1,\ldots,c_n)^t)+ (s_1,\ldots,s_n)^t ,
\end{align*}
where $T^\cW$ is defined in~(\ref{eq:TW}).
\begin{lemma}
\label{lem:involutionDW}
For each $\bfx =(x_1,\ldots,x_n)^t \in \mathbb{Z}^n$,
$${\bf \tilde{T}}^\cW ({\bf \tilde{T}}^\cW (\bfx)) = \bfx .$$
\end{lemma}

\begin{proof}
Let $\bfx =(x_1,\ldots,x_n)^t \in \mathbb{Z}^n$ be a point such that
$(x_1,\ldots,x_n)=(c_1,\ldots,c_n) + (s_1,\ldots,s_n)$,
where $(c_1,\ldots,c_n)^t \in \Lambda^\cW$ and
$(s_1,\ldots,s_n)^t \in \dS$. By the definition
of ${\bf \tilde{T}}^\cW$ we have that
\begin{align*}
& {\bf \tilde{T}}^\cW ({\bf \tilde{T}}^\cW ( \bfx) )
= {\bf \tilde{T}}^\cW (T^\cW((c_1,\ldots,c_n)^t) + (s_1,\ldots,s_n)^t)~.&
\end{align*}
Since $(c_1,\ldots,c_n)^t \in \Lambda^\cW$ it follows
by Theorem~\ref{thm:kernel_transW}
that $T^\cW((c_1,\ldots,c_n)^t) \in \Lambda^\cW$
and hence by the definition
of ${\bf \tilde{T}}^\cW$ we have that
\begin{align*}
& {\bf \tilde{T}}^\cW (T^\cW((c_1,\ldots,c_n)^t) + (s_1,\ldots,s_n)^t) & \\
& = T^\cW(T^\cW((c_1,\ldots,c_n)^t) ) + (s_1,\ldots,s_n)^t~.&
\end{align*}
Finally, by Lemma~\ref{lem:s_involutionW} we have that
\begin{align*}
& T^\cW(T^\cW((c_1,\ldots,c_n)^t)) + (s_1,\ldots,s_n)^t & \\
& = (c_1,\ldots,c_n)^t + (s_1,\ldots,s_n)^t = (x_1,\ldots,x_n)^t.&
\end{align*}
Thus,
$${\bf \tilde{T}}^\cW ({\bf \tilde{T}}^\cW ( \bfx )) = \bfx .$$
\end{proof}

\begin{theorem}
\label{thm:Dinscribed}
Let $\rho$ be the covering radius of the lattice $\Lambda^\cW$.
A Lee sphere with radius $R$ is inscribed after
the transformation~${\bf \tilde{T}}^\cW$,
inside an $n$-dimensional cube of size
$$\left(2 \left\lfloor \frac{ R+\rho }{D} \right\rfloor
+2\rho+1\right)\times \cdots \times \left(2\left\lfloor\frac{ R+\rho }{D} \right\rfloor+2\rho+1\right).$$
\end{theorem}

\begin{proof}
A Lee sphere with radius $R$ and center
$\bfc = (c_1,\ldots,c_n)^t \in \mathbb{Z}^n$, $S_{n,R} ( \bfc )$,
is the following set of points
$$
S_{n,R} (\bfc) =\left\{\bfc+ (x_1,\ldots,x_n)^t
\ \Big| \ \sum_{i=1}^n|x_i| \leq R \right\}~,
$$
where $\bfx = (x_1,\ldots,x_n)^t \in \mathbb{Z}^n$.
For each $\bfc+\bfx \in\mathbb{Z}^n$ let
$s(\bfc+\bfx) \in \dS$,
be the coset leader in
the coset of $\bfc+\bfx$ with respect to the lattice $\Lambda^\cW$.
The set $S_{n,R} (\bfc)$ is transformed after the
transformation ${\bf \tilde{T}}^\cW$ into the following set
\begin{align*}
& {\bf \tilde{T}}^\cW (S_{n,R}(\bfc)) = \left\{{\bf \tilde{T}}^\cW
(\bfc+\bfx) \ \Big| \ \sum_{i=1}^n |x_i| \leq R \right\} & \\
& = \left\{T^\cW (\bfc+\bfx-s(\bfc+\bfx))+s(\bfc+\bfx)  \ \Big| \  \sum_{i=1}^n | x_i | \leq R \right\} & \\
& = \left\{\frac{\cW \bfc}{D}+\frac{\cW (\bfx -s(\bfc+\bfx))}{D} + s(\bfc+\bfx) \ \Big| \ \sum_{i=1}^n | x_i |\leq R \right\}
&
\end{align*}
Let $\frac{\cW (\bfx - s(\bfc+\bfx))}{D} + s(\bfc+\bfx) = \bfy = (y_1,\ldots,y_n)$.
Since the covering radius of the lattice $\Lambda^\cW$, defined in
Theorem~\ref{thm:kernel_transW},
is~$\rho$, it follows that $|s(\bfc+\bfx)| \leq \rho$. Then, for $1\leq i\leq n$,
\begin{align*}
& |y_i| =\frac{\sum_{j=1}^n h_{i,j}(x_j-s(\bfc+\bfx)_j)}{D} + s(\bfc+\bfx)_i & \\
& \leq \frac{\sum_{j=1}^n (| x_j | + |s(\bfc+\bfx)_j|)}{D} + s(\bfc+\bfx)_i \leq \left\lfloor \frac{R+\rho}{D} \right\rfloor +\rho .&
\end{align*}
Therefore, the set ${\bf \tilde{T}}^\cW (S_{n,R}(\bfc))$ can be located inside
a~$\left(2\left\lfloor\frac{ R+\rho }{D}\right\rfloor+2\rho+1\right)\times \cdots
\times \left(2\left\lfloor\frac{ R+\rho }{D}\right\rfloor+2\rho+1\right)$
discrete $n$-dimensional cube which contains the points of $\Z^n$
from the set
$$ \left\{\frac{\cW \bfc}{D} + \bfy \ \Big|
\ | y_i| \leq \left\lfloor \frac{R+\rho}{D} \right\rfloor
+\rho , ~ 1 \leq i \leq n~\right\}.$$
\end{proof}

The size of an $n$-dimensional Lee sphere
with radius $R$ is $\frac{(2R)^n}{n!} + O(R^{n-1})$,
when $n$ is fixed and $R\longrightarrow \infty$.
The size of the inscribing $n$-dimensional
cube is $\left(2 \left\lfloor \frac{ R+\rho }{D}
\right\rfloor +2\rho+1\right)^n$. Since the
covering radius $\rho$ of the code~$\C^\cW$ is a low degree
polynomial in $n$ (see the next paragraph), and $n$ is fixed, we get that for
$R\longrightarrow \infty$ the size of the inscribing
$n$-dimensional cube is $\frac{(2R)^n}{n^{n/2}} + O(R^{n-1})$.
Therefore, the size of the cube is greater roughly $\frac{n!}{n^{n/2}}$ times than the size
of the $n$-dimensional Lee sphere. This is a significant improvement with respect to the
transformation given in\cite{EtYa09}, where the $n$-dimensional Lee sphere is inscribed inside
an $n$-dimensional cube of size $\frac{(2R)^n}{2^{n-1}} + O(R^{n-1})$.

Generally, it is straightforward to show that the covering radius
of the lattice $\Lambda^\cW$,
where $\cW$ is a weighting matrix of order $n$ and weight $w=D^2$,
is at most $\frac{n \cdot \sqrt{w}}{4}$, but we believe it is
considerably smaller. If $\cW$ is $\cH_m$ then an analysis
of Theorem~\ref{thm:rmj} implies that the covering radius
of $\C^{\cH_m}$ is at most $n^{1.085}$. But, we mentioned that we believe that the
covering radius is much smaller, mainly since we think that the
bound of Theorem~\ref{thm:rmm} can be improved, while we conjecture
that the bound of Theorem~\ref{thm:rmj} is quite tight.

\section{On the Existence of Weighing Matrices}
\label{sec:exist}

We have discussed weighing matrices in this paper and especially
symmetric and skew symmetric ones. Their existence
and their properties, especially their covering radius,
are important to apply the results of the paper. Of special interest
are weighing matrices for which their weight is a square.
Finally, to apply Lemma~\ref{lem:Hinscribed} and Theorem~\ref{thm:Dinscribed}
effectively, the weight of the weighing matrix should be close to its order.

It is well
known~\cite{McSl77} that if a Hadamard matrix of order $n$ exists
then $n= 1,~2$ or $n \equiv 0 (\bmod~4)$. It is conjectured that
a Hadamard matrix of order $n$ exists for each $n$ divisible by~4. There are many
constructions for Hadamard matrices and they are known to exist
for many values in this range. The first value in this range for
which no Hadamard matrix is known yet is $n=668$~\cite{KhTa04}.

In the construction of the continuous transformation,
symmetric and skew symmetric matrices are required.
In the construction
of the discrete transformation we will need a symmetric weighing matrix
whose order is a square. There are several constructions
which yield symmetric Hadamard matrices.
The matrix $\cH_m$ of order $2^m$ is a symmetric matrix.
One of Paley's constructions for Hadamard matrices yields a symmetric
Hadamard matrix of order $2(q+1)$, for each $q$ which is a power
of a prime such that $q \equiv 1 (\bmod~4)$.
This construction covers several orders which are
squares and not powers of 2, such as $n=36$, 100, 196, and
484. But, other values such as $n=144$, 324, 400, and 576, are
not covered by this construction. Another construction is given
in~\cite{MuXi06} for all orders of the form $4m^4$, where
$m$ is odd. It covers a large set of values, for example $n=324$
is covered by this construction.

We now turn our attention to conference matrices.
It is well
known~\cite{CoDi07} that if a conference matrix of order $n$ exists
then $n \equiv 0 (\bmod~2)$. If $n \equiv 0 (\bmod~4)$
then the matrix can be made skew symmetric and
if $n \equiv 2 (\bmod~4)$ it can be made symmetric.
It is conjectured that
a conference matrix of order~$n$ exists for each~$n$ divisible by 4.
If $n \equiv 2 (\bmod~4)$ then a necessary condition for the existence
of a conference matrix of order~$n$ is that $n-1$
can be represented as a sum of two squares.
It is conjectured that this condition is also sufficient.
More information on orders of other conference matrices and weighing
matrices in general can be found in~\cite{CoDi07,GeSe79}.

In general, a weighting matrix of odd order~$n$ and weight~$w$
implies that~$w$ is a square. An infinite family in this context
are weighing matrices of order $q^2+q+1$ and weight $q^2$,
for each $q$ which is a power of a prime~\cite{SWW75}.

When the discrete transformation was discussed we have considered
the code $\C^\cW_{\sqrt{w}}$ for a given weighing matrix~$\cW$ with
weight~$w$.
We note that the code $\C^\cW_s$, $s < w$, defined
in Theorem~\ref{thm:kernel_transW} seems to be not interesting
when~$\cW$ is a weighing matrix of order $n$ and weight $w < n$.
The reason for this is that except for the fact that the code is reduced
to a code in the Lee metric over the alphabet $\Z_s$, the
code essentially equals to the code $\C^\cW_w$.
All the codewords of $\C^\cW_s$ are contained in
$\C^\cW_w$. Even so we have considered the code $\C^\cW_{\sqrt{w}}$,
for the discrete transformation, since the covering
radius of the code makes it still attractive for the discrete transformation.

%@@@@@@@@@@@@@@@@@@@@@@@@@@@@@@@@@@@@@@@@@@@@@@@@@@@@@@@@@@@@@@@@@@@@@@@@%
%                                                                        %
%   6. Conclusion and Open Problems                                      %
%                                                                        %
%@@@@@@@@@@@@@@@@@@@@@@@@@@@@@@@@@@@@@@@@@@@@@@@@@@@@@@@@@@@@@@@@@@@@@@@@%
\section{Conclusion and Problems for Future Research}
%\vspace{-.25ex}
\label{sec:conclude}

We have considered the linear span of weighing matrices
as codes in the Lee and the Manhattan metrics. We have proved that the
minimum Lee distance of such a code is equal to the weight
of a row in the matrix. A set of codes related to Sylvester Hadamard matrices
were defined. Properties of these codes, such as their size, minimum
distance, and covering radius were explored.
We have defined a transformation which transforms any Lee sphere
in the space (also a conscribed cross-polytope in the continuous space) into a shape
with the same volume (in the continuous space)
located in a relatively small cube.
The transformation was defined as one of sequence
of transformations which yield a sequence of error-correcting
codes in the Lee metric. These codes are related
to the codes obtained from Sylvester type Hadamard matrices.
Many interesting questions arise from
our discussion, some of which were already mentioned.
The following questions summarize all of them.

\begin{enumerate}
\item What is the covering radius of the code
obtained from a Hadamard matrix?

\item Is the code of length $p+1$, $p$ prime, obtained
from a conference matrix of order $p+1$ is an MDS code?

\item Are the code $\C(m,j)$ and the related code $\C^{\cH_m}_{2^j}$,
equal?

\item Determine the size of the code $\C^\cH_s$ obtained from
a general Hadamard matrix $\cH$ of order $n$.

\item What is the covering radius of the code $\C^{\cH_m}_{2^j}$?

\item Is it possible to find a volume preserving transformation which transfers
each Lee sphere into a shape inscribed in a cube whose size is smaller than the
one given in our constructions? What about the same question for
a conscribed cross-polytope?

\item What is the minimum Lee distance and the covering radius
of $\C^\cW_s$, for a given weighing matrix $\cW$? The first
interesting cases are when
$\cW$ is a Hadamard matrix or a conference matrix.
\end{enumerate}

%
%*************************** References *********************************
%

\end{document}